\documentclass[10pt]{article}

\usepackage[margin=0.7in]{geometry}
\usepackage[T1]{fontenc}

\usepackage{amsthm}
\usepackage{amsmath}
\usepackage{graphicx}
\usepackage{amssymb}
\usepackage{epstopdf} 
\usepackage[square,numbers,sort&compress]{natbib}
\newtheorem{Theorem}{Theorem}
\usepackage{hyperref}
\usepackage{mathtools, amsfonts, amssymb, amsthm, bm, commath}
\usepackage{ float, wrapfig, tabularx, booktabs, subfig, listings}

\newtheorem{theorem}{Theorem}

\newtheorem{proposition}{Proposition}
\newtheorem{lemma}{Lemma}
\newtheorem*{proposition*}{Proposition}

\usepackage{cancel}
\def\R{{\cal R}}

\usepackage{caption}

\title{Dormancy stabilizes non-transitive competitive dynamics}

\author{\small Jos\'e Chac\'on$^{a,c}$, Adri\'an Gonz\'alez-Casanova$^{a,c,*}$, Imanol Nu\~nez$^{b}$, Rafael Pe\~na-Miller$^{d}$, Jos\'e Luis P\'erez$^{b}$, Johnny Yang$^{a,c}$ \\[0.8em] \small $^{a}$Center for Mechanisms of Evolution, Biodesign Institute, Arizona State University \\ \small $^{b}$Centro de Investigaciones en Matem\'aticas \\ \small $^{c}$School of Mathematical and Statistical Sciences, Arizona State University \\ \small $^{d}$Center for Genomic Sciences, Universidad Nacional Aut\'onoma de M\'exico, 62210, Mexico \\[0.5em] \small $^{*}$Corresponding author: \texttt{agonz591@asu.edu} }

\date{} 

\begin{document}
\maketitle
\begin{abstract}
Competitive interactions can maintain diversity, yet coexistence is often fragile in well-mixed populations, where stochastic fluctuations can lead to extinction. This is the case in non-transitive systems, such as rock–paper–scissors dynamics, where no single type dominates globally. While spatial structure can stabilize these systems by providing refuges in space, it remains unclear whether analogous mechanisms can operate in time in well-mixed environments.
Here, we develop a population-genetic framework showing that dormancy can act as a temporal refuge, preserving lineages and preventing collapse to fixation under interaction-driven fluctuations. We introduce a discrete-time Wright-Fisher model that combines generalized seed-banks with frequency-dependent interactions, allowing individuals to inherit their type from potential parents sampled across multiple past generations. This construction provides a tractable framework in which dormancy stores and later reintroduces lost types.
In the case of either weak or moderate selection, we prove a multidimensional diffusion limit for the resulting type-frequency process and use it to analyze complex selective interactions. In non-transitive systems, dormancy stabilizes trajectories that would otherwise collapse through stochastic extinction, extends fixation times, and sustains coexistence. These effects cannot be explained solely by an increase in effective population size.
Our results show that dormancy introduces temporal memory that qualitatively alters competitive dynamics, stabilizing otherwise fragile systems and enabling long-term coexistence.
\end{abstract}


\twocolumn

\section{Introduction}
The maintenance of diversity under strong competition is a central problem in ecology and evolution \cite{chesson2000,chesson1981}
. Although many interaction structures can in principle support coexistence, these outcomes are often not sustained in small populations, where stochastic fluctuations can lead to extinction \cite{Reichenbach2006}. Even in environments that are constant at the abiotic level, interactions among individuals can generate effective environmental fluctuations by modulating growth and mortality rates over time. Such interaction-driven variability arises across a wide range of systems, including cyclic mating strategies in animals \cite{sinervo1996rock}, non-transitive ecological interactions \cite{kerr2002}
, and antagonistic interactions in microbial communities \cite{czaran2002,granato2019}. 
From the perspective of a focal lineage, the environment is therefore endogenous and frequency-dependent. 

The consequences of these interaction-driven fluctuations depend on the structure of competition. In transitive systems, such fluctuations reinforce competitive exclusion, leading to the loss of diversity \cite{chesson2000}. In contrast, in non-transitive systems, outcomes depend on context and no single type dominates globally. 
Cyclic dominance systems, such as rock-paper-scissors interactions, are the simplest example of this structure \cite{sinervo1996rock, 
hofbauer1998}. 
Despite their potential to maintain diversity, stochastic fluctuations can eliminate one type, ultimately leading to fixation of a single dominant type \cite{Reichenbach2006}.

Previous studies have identified several mechanisms that can stabilize non-transitive interactions. Spatial structure is a canonical explanation, stabilizing competitive dynamics by creating local refuges that protect rare types from stochastic extinction \cite{kerr2002,reichenbach2007}. Furthermore, negative frequency-dependent selection can generate rare-type advantages that maintain polymorphisms \cite{sinervo1996rock}. 
Trade-offs in life-history or interaction costs, such as those underlying toxin production and resistance, can reshape competitive dynamics and support coexistence \cite{kerr2002,czaran2002}. Mutation, phenotypic switching, and behavioral variability can continually reintroduce strategies and prevent absorption, even when deterministic models predict extinction \cite{traulsen2005}. More broadly, eco-evolutionary feedbacks and higher-order interactions can modify interaction networks over time, expanding the conditions under which diversity is maintained \cite{laird2006,allesina2011,Grilli2017}.
\smallskip

In this work we propose that dormancy provides a distinct stabilizing mechanism that operates through temporal memory, analogous to spatial refuges in structured populations. By allowing individuals to temporarily exit the active population and re-enter later, dormancy stores lost types and reintroduces them into future dynamics \cite{jones2010dormancy, shoemaker2017evolution}. 
This behavior is widespread across biological systems, including seed banks in plants \cite{leck1989ecology},
spores in microorganisms \cite{plante2023spore}, persister-like states in microbial \cite{balaban2004} and cancer cell populations \cite{phan2020dormant}, and can be formalized through seed-bank models in which a fraction of individuals enters a reversible dormant state and reemerges after many generations \cite{kaj2001coalescent, blath2016new}. 
Building on Wright–Fisher models of seed-banks and complex ecological interactions \cite{kaj2001coalescent,  blath2016new, gonzalez2018duality, casanova2020lambda, CorderoHummelSchertzer2022}, we develop a population-genetic model that integrates dormancy in frequency-dependent interactions. We show that temporal memory qualitatively alters non-transitive competitive dynamics, stabilizing communities that would otherwise collapse to fixation under stochastic fluctuations, thereby shaping community composition and sustaining diversity.

\section{A Wright Fisher with  selective 
rules and a seed-bank}

We now describe a Wright-Fisher model with complex selective interactions among multiple types and that allows individuals to choose their ancestors from multiple generations in the past. Each individual in generation $n+1$ chooses one or more \emph{potential parents} from earlier generations- The generation of each potential parent is sampled independently according to a prescribed distribution, which determines the weights assigned to the immediate, intermediate, and distant past. Then, the type of each potential parent is sampled uniformly from the population in the corresponding generation. Finally, the individual's type is determined by a \emph{coloring rule} over those chosen potential parents. Formal definitions of the model, including the concepts of potential parents and coloring rules, are further detailed in the Supplementary Information (SI).
\medskip

Here, we present several examples to illustrate the model and its flexibility. 
Starting with the simplest scenario, one or two potential parents without a seed bank, we progressively extend the model to multiple potential parents and, ultimately, to the inclusion of a seed bank. The population size is denoted by $N$ and the type space by $[K]$.

\subsubsection*{Case 1: One or two potential parents and no seed-bank}

Suppose that an individual samples exactly one parent with probability $Q_N(1)$ and two parents with probability $Q_N(2)=1-Q_N(1)$. If a single parent is sampled, the offspring copies its type. If two potential parents of types $i,j\in [K]$ are sampled, the coloring rule specifies the outcome of their interaction. For example, a coloring rule may select the larger of the two labels, so that if the potential parents are of types $2$ and $5$, the offspring is of type $5$. 
This is generalized with a \emph{tournament rule}: if both parents are of the same type, the offspring inherits that type; if they are different, the offspring inherits the preferred type according to a preference relation~$\succ$. 
For instance, when there are three types, two important examples of preference rules are the
\begin{itemize}
    \item transitive rule $$1\succ 2, 2\succ 3, 1\succ 3$$ 
    \item rock-paper-scissors rule
    \begin{equation*}
    1\succ 3, 3\succ 2, 2\succ 1
    \end{equation*}
\end{itemize}

\subsubsection*{Case 2: General interactions with two potential parents and no seed-bank}

Now assume that when an individual samples two potential parents with types $\mathbf{j}=(j_1,j_2)$, the coloring rule is given by
\[
c^N_{\mathbf{j}}=(c_{\mathbf{j}}^N(1),c_{\mathbf{j}}^N(2),...,c_{\mathbf{j}}^N(K)),
\]
where $c_{\mathbf{j}}^N(i)$ denotes the probability that the offspring is of type $i$, given potential parents of types $\mathbf{j}$. If the frequency of individuals of type $j$ in the previous generation is $z_j$, then the probability that the offspring is of type $i\in [K]$ is
\[
    p_i^N(z)
    := z_i Q_N(1)
       + Q_N(2)\sum_{\mathbf{j}\in[K]^2} z_{j_1} z_{j_2}\, c_{\mathbf{j}}^N(i).
\]
The tournament case is recovered by considering coloring rules that take values in $\{0,1/2,1\}$. Another example is the rule
\[
c^N_{j_1,j_2}(\min(j_1,j_2))
=1-c^N_{j_1,j_2}(\max(j_1,j_2))
=\frac{1}{|j_1-j_2|+1},
\]
which favors larger labels, with the advantage increasing in the distance between types.

\subsubsection*{Case 3: General interactions with many potential parents and no seed-bank}

The construction extends naturally to any number $m$ of potential parents. If an individual samples $m$ parents with types $\mathbf{j}=(j_1,\ldots,j_m)$, the coloring rule $c^N_{\mathbf{j}}$ determines the offspring type. The transition probabilities are
\begin{equation}\label{types}
    p_i^N(z)
    = \sum_{m=1}^\infty Q_N(m)
      \sum_{\mathbf{j}\in[K]^m}
      \Bigg( \prod_{\ell=1}^m z_{j_\ell} \Bigg)\,
      c_{\mathbf{j}}^N(i) \ ,  \ i\in [K].
\end{equation}
This general formulation allows for a wide range of interaction mechanisms, such as:

\medskip
\begin{itemize}
    \item Majority rule: The offspring adopts the most frequent type among its potential parents, with ties broken uniformly:
    \\
\\    $  c_{\mathbf{j}}^N(i)=$
   { \small  \begin{equation}\label{majority}
        \begin{cases}
            1, & \text{if $i$ is the majority type in } (j_1,\ldots,j_m),\\
            1/d, & \text{if $i$ is among the $d$ most frequent types},\\
            0, & \text{otherwise.}
        \end{cases}   
    \end{equation} }
    This rule captures situations in which heterozygous phenotypes are disadvantaged relative to homozygous ones.
    \item Minority rule: The offspring adopts the least frequent type among the potential parents, again with uniform tie-breaking. This rule models balancing selection.
\end{itemize}

\subsection*{Seed-bank: Sampling potential parents from the past}
We now incorporate a seed-bank such that when sampling a potential parent, an individual selects from the previous generation with probability $r_N$, and from deeper generations in the past with probability $1-r_N$. The generation offset is geometrically distributed with parameter $q_N$, so that older generations are chosen with exponentially decreasing probability.
\medskip

Let $X_i^N(n)$ equal 
\[
    \frac{\text{number of type $i$ individuals in generation $n$}}{N},
\]
which is the frequency of type $i$ at generation $n$. Since the process $(X^N(n))_{n\ge 0}$ is not Markovian, we introduce a \emph{memory variable}
\[
    Y_i^N(n) :=
    \sum_{j=1}^\infty q_N(1-q_N)^{j-1} X_i^N(n-j),
\]
which satisfies the recursion
\[
    Y^N(n+1)=q_N X^N(n)+(1-q_N)Y^N(n).
\]

Thus, each individual effectively samples parents from the mixed distribution
\[
    r_N X^N(n)+(1-r_N)Y^N(n),
\]
which combines the current population with a weighted memory of the past.

\subsubsection*{Updating the population}
Given the parental-type distribution
\[
    z = r_N X^N(n)+(1-r_N)Y^N(n),
\]
the next generation is obtained via $N$ independent draws from $p^N(z)$ in ~\eqref{types}:
\[
    X^N(n+1)
    \sim \frac{1}{N}\mathrm{Multinomial}\big(N,\,p^N(z)\big).
\]
This provides an equivalent Markovian description of the  type-frequency process.

\subsection*{Scaling Limits}

Now we let the population size $N$ tend to infinity and derive scaling limits for the type-frequency process.\medskip

We work under the main assumption 
{\small $\rho_N := 1 - Q_N(1) \to 0,$} and show that
$
(X(t)) := \lim_{N \to \infty} X^N(\lfloor \rho_N t \rfloor)
$
converges to a system of ordinary differential equations (deterministic limit) if
$(N\rho_N)^{-1} \to 0$, and to a system of stochastic differential equations (stochastic limit) if\\
$(N\rho_N)^{-1}  \to \text{constant}$.

\begin{Theorem}
\label{T1}
\small 
Consider a population with $K$ types, $\{1,2,\dots,K\}$, and a coloring rule
$c^N_{\mathbf{j}}$. Let $\mu:\mathbb{R}^K \to \mathbb{R}^K$ be such that for
$x \in \Delta_K$, $\mu(x) = (\mu_1(x),\dots,\mu_K(x))$ is defined by
{\small
\begin{equation}\label{mu}
\mu_i(x)=\rho_N^{-1}\left(
\sum_{m=1}^\infty Q_N(m)
\sum_{\mathbf{j} \in [K]^m}
\left( \prod_{\ell=1}^m x_{j_\ell} \right)
c_{\mathbf{j}}^N(i)
- x_i\right)
\end{equation}
}

Define the matrix $\zeta(x) = (\zeta_{ij}(x))_{1 \le i,j \le K}$ by
\[
\zeta_{ij}(x) :=
\begin{cases}
\qquad \quad 0, & \text{if } i < j,\\[1pt] 
\sqrt{
\frac{x_i \bigl(1 - \sum_{k=1}^i x_k\bigr)}
     {1 - \sum_{k=1}^{i-1} x_k}
},
& \text{if } i = j, \\[1pt]
- x_i
\sqrt{
\frac{x_j}
{\bigl(1 - \sum_{k=1}^{j-1} x_k\bigr)
 \bigl(1 - \sum_{k=1}^j x_k\bigr)}
},
& \text{if } i > j.
\end{cases}
\]
Then $(X(t)) = \lim_{N \to \infty} X^N(\lfloor \rho_N t \rfloor)$ is the unique strong solution
of the system of stochastic differential equations
{ \small
\begin{equation}
\label{sde_limit}
\begin{split}
dX(t)
&= \mu(X(t))\,dt
+ \sqrt{\sigma}\,\zeta(X(t))\,dB_t
+  \alpha\bigl(Y(t)-X(t)\bigr)\,dt, 
\\
dY(t)
&=  \beta\bigl(X(t)-Y(t)\bigr)\,dt,
\end{split}
\end{equation}
}
where $\lim_{N\to \infty}(N \rho_N)^{-1}= \sigma$, $\lim_{N\to \infty}(1-r_N)/\rho_N= \alpha$, $\lim_{N\to \infty}=q_N/\rho_N=\beta$ and $B$ is a $K$-dimensional Brownian motion. 
\end{Theorem}

Theorem~1 shows that, under a natural scaling, the Wright-Fisher model with complex selective interaction rules and seed-bank converges to a coupled diffusion system in which selection, stochastic drift, and dormancy 
jointly shape the long-term evolutionary dynamics. 

\section{Results}

Three distinguished examples of evolutionary games considered within our framework are the canonical rock-paper-scissors selective interaction, tournaments extending this type of interaction to many types, and the majority voting rule. 
Memory has an important effect on each of these games, as will be explained in the rest of this section.
\medskip

Hereafter, for the interpretation of simulations and figures, $N$ denotes the population size and $\sigma$ is understood as the reciprocal of the effective population size. The different seed bank regimes are defined through $\alpha$, the average number of individuals in the present population descending from a given individual in the distant past, and $\beta N$, the typical number of generations separating such individuals from their ancestors.

\subsection{Rock-Paper-Scissors}
We analyze the effect of sustaining a seed bank when three types of individuals compete under rock-paper-scissors (RPS) dynamics, under nine combinations of various population sizes and seed bank regimes (see Figure 1). A graphical representation of the RPS rule and of the corresponding model are presented in Fig.~\ref{fig:FixingTimes}A and Fig.~\ref{fig:FixingTimes}B, respectively. Simulations are based on the scaling limit of Theorem \ref{T1}. Further details of both the discrete model and its scaling limit are provided in the Methods section and described more formally in the SI.

\begin{figure*}[!t]
    \centering\includegraphics[width=12cm]{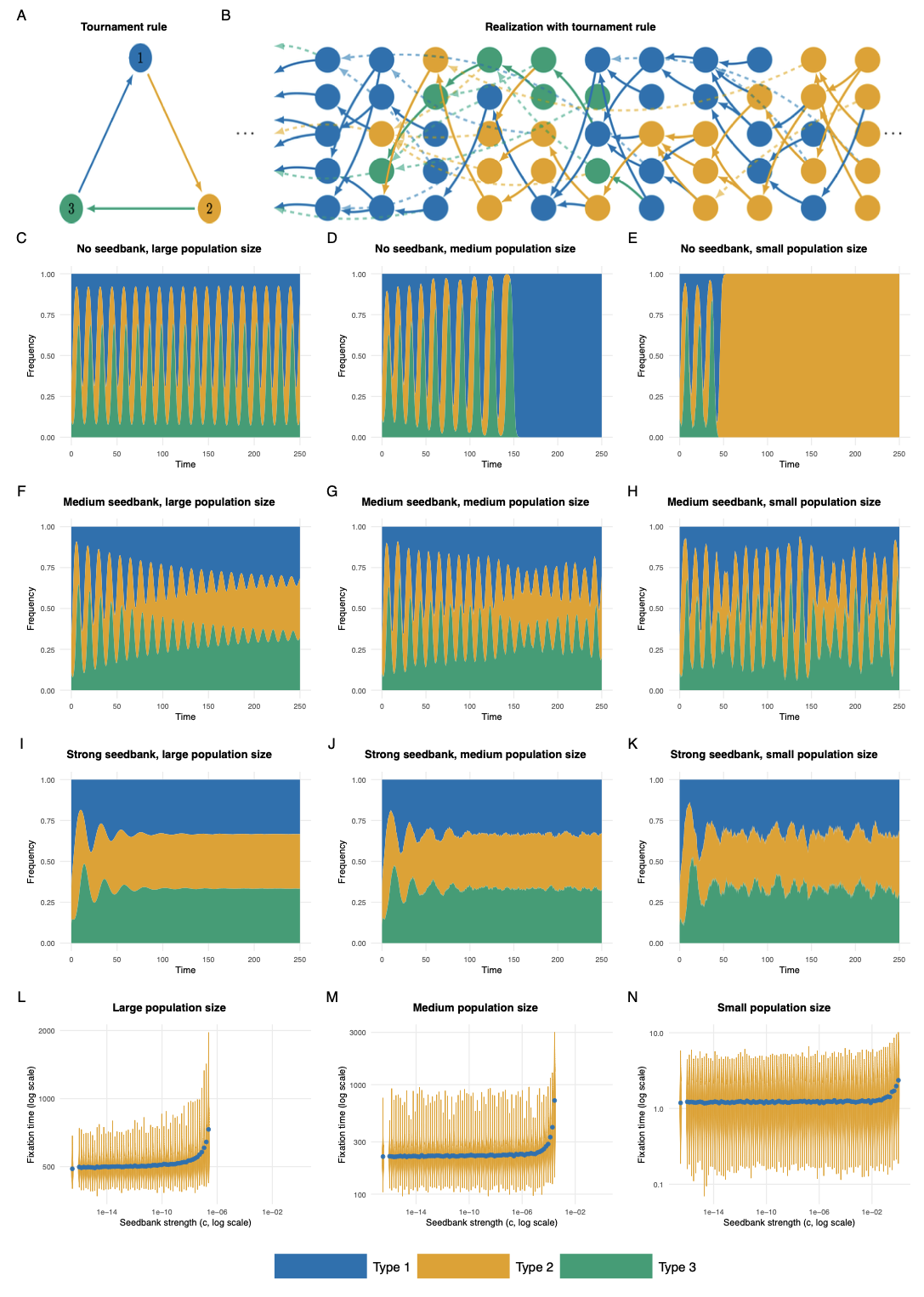}

        \caption{ \small (A) Schematic representation of the tournament rule for three types of individuals. An arrow from $i$ to $j$ (drawn in the color of vertex $j$) indicates that a selective interaction between a type~$i$ and a type~$j$ individual results in an offspring of type~$j$.
(B) Graphical representation of the model with a seed bank, incorporating the tournament rule shown in panel~A, in the case where individuals have one or two potential parents.
(C–K) Simulations of the scaling limit \eqref{sde_limit} for the three type system under varying seed-bank strength $c$ (with $\alpha = \beta = c$) and diffusion intensities $\sigma$. Specifically: no seed bank ($c=0$); medium seed bank ($c=10^{-2}$); strong seed bank ($c=1$). 
Reciprocal of effective population sizes correspond to $\sigma^2=10^{-6}$ (large effective population), $\sigma^2=3\times10^{-4}$ (medium effective population), and $\sigma^2=2\times10^{-3}$ (small effective population).
(L–N) Fixation times from $1{,}000$ simulations for each seed-bank strength $c$ (log–log scale), for the same three values of $\sigma$ as above. Violin plots show the empirical distributions, and dots indicate mean fixation times across replicates. 
}
    \label{fig:FixingTimes}
\end{figure*}

\medskip

First, we consider the classical scenario in which there is no seed bank and the population is large, leading to stable orbits where the most prevalent type changes periodically (Fig.~\ref{fig:FixingTimes}C).  Moving to the right columns in Figure \ref{fig:FixingTimes}, the population size decreases, leading to a stronger random genetic drift. This reveals the fragility of the cyclic dynamics: if, due to a random fluctuation, one type becomes extinct, the system reduces to a two-type model with directional selection, eventually resulting in fixation of a single type. This occurs with a medium population size after approximately 150 generations (Fig.~\ref{fig:FixingTimes}D), and with a small population size after fewer than 50 generations (Fig.~\ref{fig:FixingTimes}E).
\medskip

As we move to the lower rows (Fig.~\ref{fig:FixingTimes}F–K), the model is extended to include a seed bank.
In the middle row (Fig.~\ref{fig:FixingTimes}F–H), the average germination time is moderate, while in the bottom row (Fig.~\ref{fig:FixingTimes}I-K) it is large relative to the noise. The first striking result is that as the strength of the seed bank increases, the system becomes more resilient and fixation no longer occurs in the first 250 generations. To test the hypothesis that stronger seed banks lead to longer fixation times, we simulated fixation times for different seed bank parameters under three population-size regimes. In all cases (Fig.~\ref{fig:FixingTimes}L-N), we observe that seed banks dramatically extend fixation times. 
\medskip

While the previous effect is partly explained by the increase in the effective population size, since a seed bank increases the number of potential parents per individual, 
Figure~\ref{fig:FixingTimes} makes clear that the effect of a strong seed bank is qualitatively different from simply increasing the population size. Observe that in Fig.~\ref{fig:FixingTimes}C, with a large effective population size due to the given population size, we observe orbits.
In contrast, in Fig.~\ref{fig:FixingTimes}K, 
the strong seed bank leads to an increased effective population size, and we observe a stable quasi-equilibrium.
\medskip

For the RPS rule in the absence of a seed bank, the unique Nash equilibrium is $$\left(\frac{1}{3}, \frac{1}{3}, \frac{1}{3}\right)$$ and is not an attractor; instead, the dynamical system exhibits cyclical behavior. Remarkably, the presence of noise and a seed bank drastically alters this behavior, as shown in Figure~1. 

\subsubsection*{Moderate selection and strong coexistence}
In the particular case where the probability that an individual samples more than one potential parent tends to zero slower than the inverse of the total population size (moderate selection), the scaling limit in Theorem \ref{T1} is given by a  system of ordinary differential equations.  
\newpage
More precisely, in this regime one obtains a limiting system of the form 
\begin{equation*}
\begin{split}
x'(t)
&= \mu(x(t))\,
+  \alpha\bigl(y(t)-x(t)\bigr), \\
y'(t)
&=  \beta\bigl(x(t)-y(t)\bigr),
\end{split}
\end{equation*}
where $\mu$ is given in \eqref{mu}. 

For the RPS game, the limiting dynamics are therefore 
\begin{equation*}
\begin{aligned}
x'_1 (t)&= x_1(t)(x_3(t)-x_2(t))+\alpha(y_1(t)-x_1(t)),\\
x'_2(t) &= x_2(t)(x_1(t)-x_3(t))+\alpha(y_2(t)-x_2(t)),\\
x'_3(t) &= x_3(t)(x_2(t)-x_1(t))+\alpha(y_3(t)-x_3(t)),\\
y'_1(t) &= \beta(x_1(t)-y_1(t)),\\
y'_2(t) &= \beta(x_2(t)-y_2(t)),\\
y'_3(t) &= \beta(x_3(t)-y_3(t)),
\end{aligned}
\end{equation*}
where $x_i$ denotes the frequency of type $i$ individuals in the active population, while $y_i$ denotes the corresponding frequency in the dormant population. The parameter $\alpha$ controls the rate at which active individuals become dormant, and $\beta$ the rate at which dormant individuals reactivate. The $x$-subsystem corresponds to the cyclic Lotka-Volterra (replicator) dynamics, while the $y$-variables relax linearly toward $x$. 
\smallskip

When $\alpha=\beta=0$, it is well known that this system arises as the deterministic limit of RPS dynamics and admits a center at $x=(1/3,1/3,1/3)$
with orbits. 
In contrast, when $\alpha,\beta>0$, the system admits a locally asymptotically stable equilibrium $x=y=(1/3,1/3,1/3)$, corresponding to the state in which all types are equally represented in both the active and dormant populations. At this equilibrium, the frequencies remain constant over time. We refer to this phenomenon as \emph{strong coexistence}, which can be caused by  the presence of a seed bank  in a broader class of evolutionary games and not only in the RPS setting.

\subsection{Tournaments}

Selective interactions within a population are often much more complex than a three-type rock-paper-scissors scenario. 
Our hypothesis is that, in the presence of many types, seed banks enhance resilience by maintaining biodiversity and generating an intermittence effect, in which population types temporarily disappear from the active population and reemerge after several generations.
\medskip

To test this hypothesis, now we focus on two tournament interactions that naturally extend the rock-paper-scissors dynamics to many types (see Fig.~\ref{fig:tournaments59}). Details of the tournament rules and simulation setup are provided in the SI, with simulations again based on the scaling limit \eqref{sde_limit}.
\begin{figure*}[!t]
    \centering
\includegraphics[width=12.5cm]{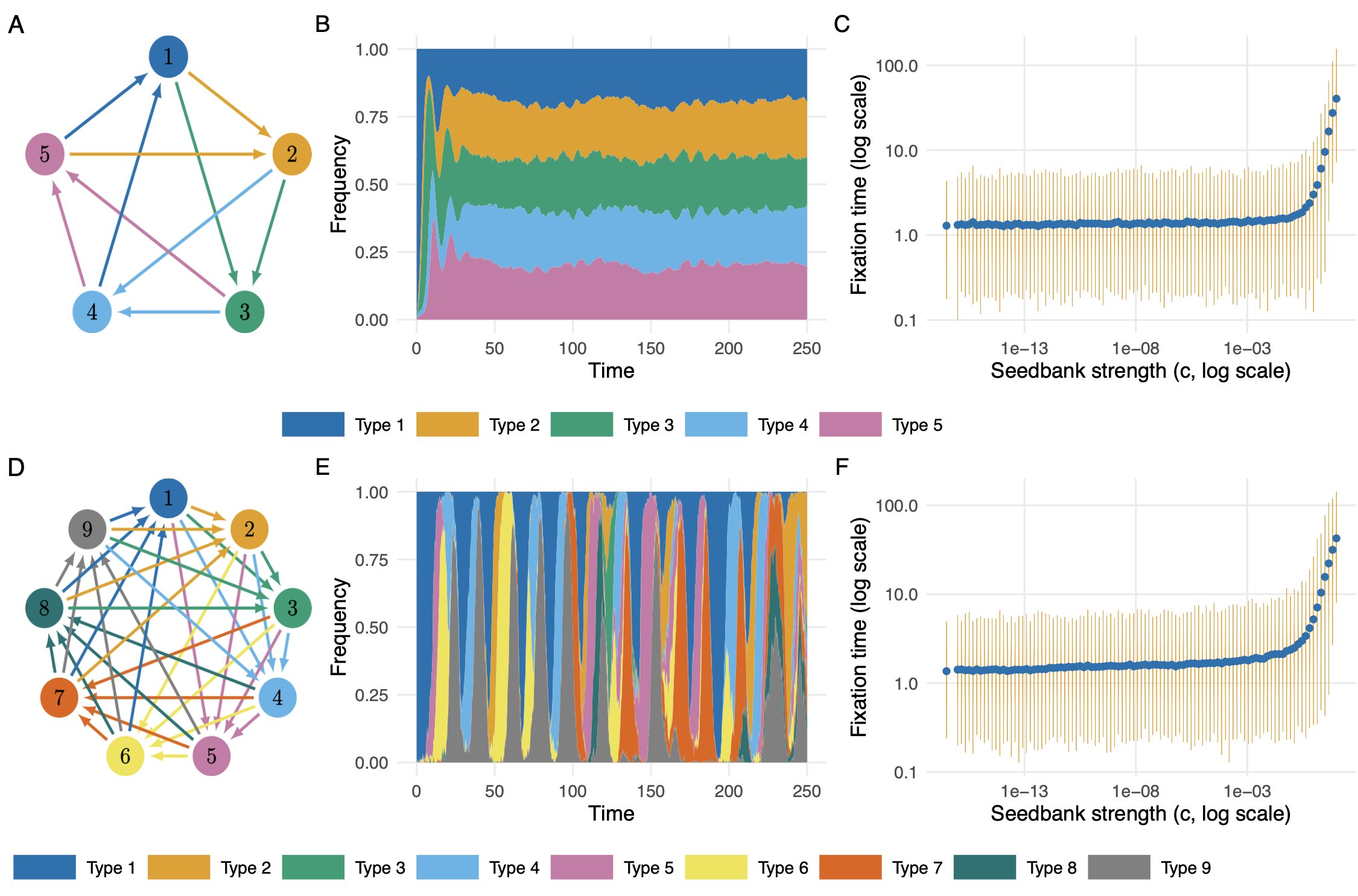}
\caption{
(A,D) Schematic representations of the tournament rule for five and nine types of individuals, respectively. 
(B,E) Simulations of the scaling limit \eqref{sde_limit} for 5- and 9-type systems, respectively. In (B), the initial conditions are 
$x_0 = (1,0,0,0,0)$ and $y_0 = (1/5, 1/5, 1/5, 1/5, 1/5)$, with seed bank strength $\alpha = \beta = 0.1$ and diffusion intensity $\sigma^2 = 10^{-4}$. 
In (E), $x_0=(1,0,\dots, 0)$ and $y_0 = (1/9, \dots, 1/9)$, with $\alpha = \beta = 0.01$ and $\sigma^2 = 10^{-2}$. 
(C,F) Fixation times from $1{,}000$ simulations for each seed bank strength $\alpha = \beta = c$ (log-log scale), for the 5- and 9-type systems, respectively. Violin plots show the empirical distributions, and dots indicate mean fixation times across replicates. 
}
    \label{fig:tournaments59}
\end{figure*}
\medskip

We first consider a population consisting of five types, where each type has a selective advantage over two types and a selective disadvantage against the remaining two (Fig.~\ref{fig:tournaments59}A).  It is shown that the presence of a seed bank leads to strong coexistence, with the frequency of each type fluctuating mildly around $1/5$ (Fig.~\ref{fig:tournaments59}B). This sharply contrasts with the case without a seed bank, where orbits emerge, frequencies deviate substantially from $1/5$, and stochastic fluctuations eventually drive one type to extinction, resulting in an irreversible loss of biodiversity (in analogy with Fig.~\ref{fig:FixingTimes}D for the three type system).
\medskip

Next we consider a population of nine types, where each type has a selective advantage against four types and a disadvantage against the remaining four (Fig.~\ref{fig:tournaments59}D). While the seed bank again promotes biodiversity, the qualitative behavior of the frequencies is markedly different. We observe large fluctuations, including periods during which a type disappears from the active population and subsequently reemerges after several generations (Fig.~\ref{fig:tournaments59}E).  
\medskip

In both cases, increasing the strength of the seed bank suppresses extinction events and thus enhances biodiversity. As in the RPS model, we illustrate this by showing that stronger seed banks lead to longer fixation times (Fig.~\ref{fig:tournaments59}C,F). 
\medskip

This analysis provides a clear visualization of how maintaining a seed bank leads to cryptic variation at the ecosystem level: the effective biodiversity is substantially larger than the observable biodiversity at any given point in time.  Together, these results demonstrate that dormancy not only preserves genetic diversity, but also fundamentally reshapes the outcomes of evolutionary games and the structure of ecological communities.

\subsection{Majority Voting}

The presence of a seed bank can also fundamentally change the behavior of competing populations under frequency-dependent selection. To illustrate this phenomenon, we study the scaling limit of our model under the majority voting rule; that is, when the type of an individual is assigned as the most common type in its potential parents pool. We make a connection to the concept of hybrid zones in nature \cite{barton1989adaptation} and the celebrated Allen-Cahn equation \cite{allen1979microscopic}.
\medskip

Etheridge et al.\ \cite{etheridge2017branching, etheridge2022wide, etheridge2024looking} considered a spatial genetic model in which each individual carries two alleles, each of type $a$ or $A$. Homozygotes (types $aa$ or $AA$) are assumed to be equally fit and selectively advantageous relative to heterozygotes (type $aA$). Under appropriate initial conditions and weak selection, on large spatial and temporal scales, the proportion $u(x,t)$ of $a$-alleles satisfies the Allen-Cahn equation
\begin{equation}\label{AC}
\frac{\partial u}{\partial t} = \Delta u + s\,u(1-u)(2u-1),    
\end{equation}
where $s>0$ is a scaled selection coefficient. In \cite{etheridge2017branching} it is shown that the process $u(x,t)$ can be constructed as the dual of a branching Brownian motion with majority rule as its coloring mechanism. In this setting,  a hybrid zone is a region of space in which both alleles are present in substantial numbers. For example, starting from an initial configuration in which the two different homozygote types occupy disjoint regions, a hybrid zone can be formed eventually through mating of the originally isolated populations.  Under suitable initial conditions, a diffusive limit of \eqref{AC} leads to the indicator function of an hybrid zone  whose boundary evolves according to mean curvature flow (see Theorem 1.3 \cite{etheridge2017branching}). 
\medskip

Under the majority voting rule \eqref{majority}, the scaling limit of our model leads to a non-spatial frequency process $u(t)$ satisfying the system of equations: 
{\small\begin{align}\label{ACSB}
    \mathrm du(t) &= s u(t)(1-u(t))(2u(t) - 1) \mathrm dt + \sqrt{\sigma u(t)(1-u(t))}\mathrm dB_t \notag
    \\[.25cm]
    & \quad + s \alpha  \left(y(t) - u(t) \right)\mathrm dt,  \notag \\ \notag
    \\
    \mathrm dy(t)&= s \beta \left(u(t) - y(t) \right)\mathrm dt,
\end{align}}

where the first term describes selection induced by the majority voting rule, the second term captures stochasticity, and the rightmost term in the first line, along with the second-line equation, accounts for seed bank effects. (see Proposition 2 in SI).
We refer to \eqref{ACSB} as the non-spatial Allen-Cahn equation {with seed bank. 
\smallskip

In this context, we define a \emph{temporal hybrid zone} as the transient interval during which the population remains effectively polymorphic. Specifically, it represents the duration of time that allelic frequencies persist at intermediate values (segregating between the absorbing boundaries of 0 and 1) before the system is driven back to the apparent fixation of either type.
\medskip

Simulations of the scaling limit \eqref{ACSB} provide remarkable insights: In the absence of a seed bank, any solution starting in $(0,1)$ but above $1/2$ typically converges to $1$, while solutions starting below $1/2$ tend to converge to $0$. The initial condition $1/2$ is stationary. The inclusion of a seed bank changes the picture completely. When the process starts from an initial condition in which both types are present, coexistence emerges. Moreover, this coexistence appears in two distinct regimes. In the case of a weak seed bank (Fig.~\ref{fig:AllenCahn}A), long periods during which the population is almost completely dominated by type $A$ are followed by periods dominated by type $a$, interspersed with short phases of hybrid over-representation that resemble the hybrid zones observed in the spatial Allen-Cahn model (see SI). On the other hand, when the seed bank is very strong, the system converges to a stable coexistence equilibrium, and the switching behavior disappears (Fig.~\ref{fig:AllenCahn}D). This is in line with the strong coexistence effect discussed before.   
\medskip

Interestingly, the discrete model exhibits an additional phenomenon: long periods of  orbital behavior can abruptly end in rapid fixation of one of the two types. This phenomenon raises interesting ecological and mathematical questions, which will be the subject of future research.
\begin{figure*}[!t]
    \centering
   \includegraphics[width=10cm]{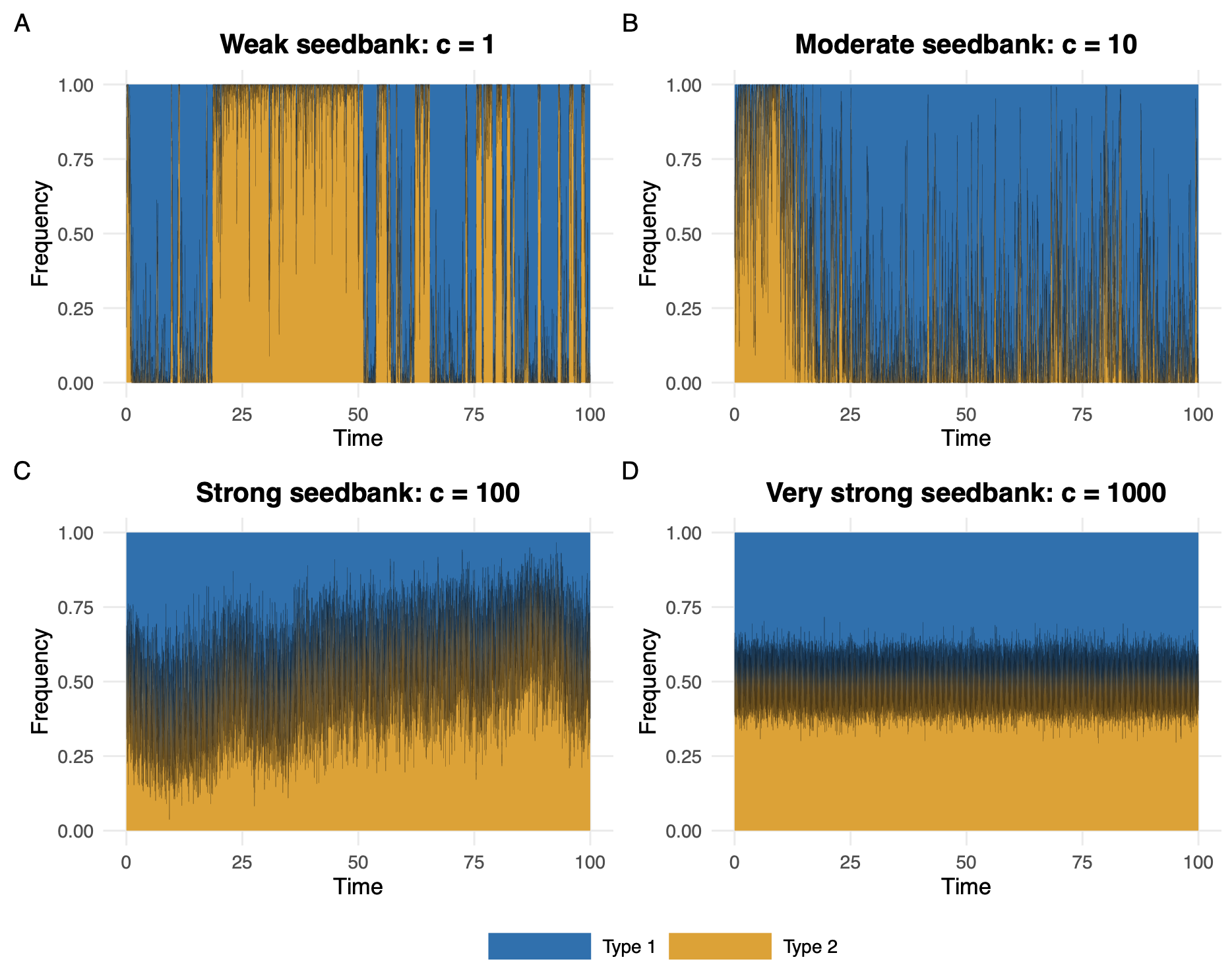}
    \caption{Effect of the seed bank strength ($\alpha = \beta = c$) on the Allen–Cahn equation with seed bank \eqref{ACSB}. 
    Each panel shows trajectories of the two-type system under increasing coupling between the active population and the seed bank. Larger $c$ strengthens memory effects, slows transitions, and stabilizes the interface between types. 
    } 
    \label{fig:AllenCahn}
\end{figure*}

\section{Discussion}

Memory is intrinsic to many evolutionary systems: it allows populations to retain and reintroduce past states, expanding the range of possible evolutionary outcomes. These effects are particularly relevant in systems governed by frequency-dependent interactions, which are naturally described by evolutionary game theory. This framework captures a wide range of ecological interactions, including antagonistic, predator-prey, and mutualistic dynamics \cite{smith1973logic, dugatkin1998gametheory, mcnamara2020gametheory, axelrod1985evolution, brown1987predator, gokhale2012mutualism}. A canonical example arises in the lizard \emph{Uta stansburiana}, where three mating strategies interact through a rock-paper-scissors structure and produce persistent population cycles \cite{sinervo1996rock}. More generally, such interaction networks can in principle maintain diversity, but often do not sustain coexistence in well-mixed populations, where stochastic fluctuations can eliminate types and lead to collapse to simpler competitive hierarchies \cite{Reichenbach2006, reichenbach2007}.
\smallskip

Dormancy provides a natural form of memory in biological systems by allowing a fraction of individuals to enter a reversible dormant state and reemerge after many generations \cite{jones2010dormancy,levin1990seed,  locey2017microscale, shoemaker2017evolution}.
This mechanism is widespread, including seeds in plants \cite{levin1990seed, Kim2026}, 
endospores in bacteria \cite{setlow2014spore,plante2023spore}, and persister-like states in cancer cells \cite{klein2011framework, friberg2015cancer, phan2020dormant}.
By storing and reintroducing past population states, dormancy creates a temporal reservoir of diversity that can influence eco-evolutionary dynamics.
\smallskip

Our results show that this temporal memory has a direct ecological consequence: it can stabilize competitive systems in the absence of spatial structure. Complex selective interactions 
and spatial refuges 
are known to maintain diversity by protecting rare types in different regions of space. In contrast, seed banks provide an analogous mechanism in time. By allowing individuals to exit and later re-enter the active population, dormancy acts as a temporal refuge that preserves lineages through periods of unfavorable ecological conditions.
\smallskip

For instance, this mechanism qualitatively alters the behavior of classical interaction structures such as rock-paper-scissors dynamics. In well-mixed populations, these systems are vulnerable to stochastic extinction and loss of diversity. In our model, dormancy suppresses these extinction events and sustains coexistence. In the deterministic limit, the dynamics admit a stable polymorphic equilibrium, a phenomenon we refer to as strong coexistence, which cannot be explained solely by an increase in effective population size.
\medskip

More broadly, dormancy generates what can be described as cryptic biodiversity: lineages may disappear from the active population and later reemerge from the seed bank, so that the effective diversity of the system exceeds what is observable at any given time. In our simulations, types repeatedly reach zero frequency yet reappear when ecological conditions become favorable. This mechanism provides a natural explanation for how diversity can be maintained in well-mixed systems where no spatial refuges are available.
\smallskip

We find that dormancy can qualitatively transform frequency-dependent selection beyond cyclic interactions. In the case of majority voting dynamics, the presence of a seed bank produces rich temporal behavior even in the absence of space, including regimes of alternating dominance and stable coexistence. In analogy with spatial models such as the Allen-Cahn equation \cite{allen1979microscopic,barton1989adaptation, etheridge2017branching}, this suggests the existence of temporal analogues of hybrid zones, in which diversity is maintained through fluctuations in time rather than space.
\smallskip

Dormancy is typically interpreted as a response to externally imposed fluctuations, such as intermittent stress or antibiotic exposure 
\cite{plante2023spore,balaban2004,phan2020dormant, setlow2014spore, kussell2005}. 
Our results provide a new perspective on this role by showing that similar selective pressures can arise endogenously through ecological interactions. Frequency-dependent interactions generate effective fluctuations in growth and mortality, creating conditions under which dormancy, often referred to as persistence in microbial systems, can promote lineage survival. By introducing a temporal dimension to evolutionary dynamics with the inclusion of seed banks, dormancy qualitatively alter the behavior of competitive systems, stabilizing interaction networks that would otherwise collapse to fixation under stochastic fluctuations. This suggests that the stabilizing effects of dormancy on eco-evolutionary dynamics could be favored by selection.

\bibliographystyle{abbrvnat}
\bibliography{references.bib}


\clearpage 

\onecolumn 

\appendix 
\setcounter{section}{0} 
\setcounter{figure}{0}  
\setcounter{table}{0}   
\setcounter{equation}{0} 

\renewcommand{\thesection}{S\arabic{section}}
\renewcommand{\thefigure}{S\arabic{figure}}
\renewcommand{\thetable}{S\arabic{table}}
\renewcommand{\theequation}{S\arabic{equation}}

\section*{SUPPLEMENTARY INFORMATION}

\section{Wright-Fisher model with complex interactions and a seed-bank}

Building on previous models with complex ecological interactions and seed banks \cite{kaj2001coalescent, blath2016new, gonzalez2018duality, casanova2020lambda, CorderoHummelSchertzer2022}, we introduce a novel Wright–Fisher model incorporating both complex interactions and a seed bank. 
\medskip

Consider a discrete-time population model of fixed size \( N \), consisting of individuals of \( K \) possible types. At each generation, each individual samples a random collection of potential parents, possibly from all previous generations, and is assigned a type according to the observed types and a prescribed \emph{coloring rule}. 

\medskip

In the following, \( [N] \) denotes the set of individual labels \( \{1, \dots, N\} \); \( [K] \) denotes the type space \( \{1, \dots, K\} \); and \( [K]^m = \{(k_1, \dots, k_m) : k_i \in [K]\} \) denotes the set of $m$-dimensional pools of types. The sampling sets for the type of an individual is then represented by
\[
\mathcal{C} = \bigcup_{m \in \mathbb{N}^*} [K]^m.
\]
In other words $\mathcal{C}$ represents the possible sets of potential parents with their type that can be sampled by an individual.
\medskip

Building on \cite{gonzalez2018duality,casanova2020lambda,CorderoHummelSchertzer2022}, a \emph{coloring rule} is a family of probability distributions over the type space \( [K] \):
\[
    C_N = \{ c_z^N(\cdot) \}_{z \in \mathcal{C}}.
\]
Lastly, the simplex over $K$ types is 
\[
    \Delta_K = \left\{(x_1, \dots, x_K) \in \mathbb{R}^K : x_i \ge 0,\ \sum_{i=1}^K x_i = 1\right\}.
\]
For each generation $n$ and individual \( v \), we impose the following reproduction mechanism:
\begin{enumerate}
    \item The number of potential parents sampled for \(v\), denoted \(K_v(n)\), is distributed according to \(Q_N\); that is,
    \[
        \mathbb{P}\big(K_v(n)=k\big) = Q_N(k), \quad k \in [N].
    \]

    \item  For each of the \( K_v(n) \) potential parents, the generation from which it originates is given by
    $$\begin{cases}
      n-1 & \text{with probability}\  r_N 
      \\
      \\ n  - 1 -j \ & \text{with probability}\ (1-r_N)(1-q_N)^{j-1}q_N \ , \ j\geq 1
    \end{cases}$$
    Conditional on the sampled generation, the potential parent is chosen uniformly among all individuals in that generation.

    \item The previous steps define a parent pool \( z \in \mathcal{C} \), and $v$ is assigned a type according to the distribution $c_z^N$:
\[
\mathbb{P}(\text{type}(v)=i \mid z) = c_z^N(i), \quad i \in [K].
\]
\end{enumerate}

The resulting system is formally described by the vector of type frequencies of the population. For each generation \( n \in \mathbb{N} \), let
\begin{equation}\label{X}
    X_i^N(n) := \frac{\#\{v \in V_n : \operatorname{type}(v) = i\}}{N},
\end{equation}
denote the frequency of individuals of type \( i \), where \( V_n \) is the set of individuals at generation \( n \). The frequency vector of types at generation $n$ is then
\begin{equation*}
      X^N(n) = \big(X_1^N(n), \dots, X_K^N(n)\big).
\end{equation*}
Remarkably, the process $(X^N(n))_{n \ge 0}$ is no longer Markovian. In order to restore Markovianity and obtain a tractable description of the population dynamics, we proceed as follows. 
\medskip

Let $(V_{j}^N(n))_{j = 1}^N$ be random variables with values in $[K]$, where $V_j^N(n)$ denotes the type of the $j$-th individual in generation $n$. In addition, define the auxiliary $\Delta_K$-valued process $(Y^N(n))_{n \in \mathbb{Z}}$ by
\begin{equation}\label{Y}
Y_i^N(n) := q_N \sum_{j=1}^{\infty} (1-q_N)^{j-1} X_i^N(n - j), \quad i \in [K],\ n \in \mathbb{Z},    
\end{equation}
which tracks the proportion of individuals of each type in past generations, weighted geometrically. The evolution of \[Y^N = (Y_1^N(n),\dots,Y_K^N(n))\] can be written as
\[
    Y^N(n + 1) = q_N X^N(n) + (1-q_N) Y^N(n), \qquad n \in \mathbb{Z}.
\]

Moreover, define $p^N : \Delta_K \to \Delta_K$ by
\begin{equation} \label{eq:defPi}
    p_i^N(z) = \sum_{n = 1}^\infty Q_N(n) \sum_{\bm{j} = (j_1,\dots,j_n) \in [K]^n}
    \prod_{l = 1}^n z_{j_l} \, c_{\bm{j}}^N(i).
\end{equation}
Recall that $Q_N(n)$ is the probability of having $n$ potential parents, and $z_{j_l}$ is the probability that a potential parent is of type $j_l$. Then, for each $j \in [N]$,
\[
    \mathbb{P}(V_{j}^N(n + 1) = i \mid X^N(n) = x, Y^N(n) = y) = p_i^N(r_N x + (1 - r_N) y).
\]

Let \( M^{(N)}(n+1) \in \mathbb{N}^d \) be the random vector representing the population composition at generation \( n+1 \). Its \( i \)-th component, denoted \( M^{(N)}_i(n+1) \), is the number of individuals of type \( i \) in generation \( n+1 \), for each \( i = 1, \dots, d \). Note that
\[
M^{(N)}(n+1) = \sum_{j = 1}^N e_{V_j^N(n + 1)}.
\]
Therefore, conditional on $\{X^N(n) = x, Y^N(n) = y\}$, 
\begin{equation}\label{multinomial}
    M^{(N)}(n+1) \sim \operatorname{Multinomial}\bigl(N,\, p^N(r_N x + (1 - r_N) y)\bigr).
\end{equation}

Hence, the frequency vector of types at time $n + 1$ is equivalently given by
\begin{equation*}
X^N(n  + 1) = \frac{M^{(N)}(n  + 1)}{N}.    
\end{equation*}

\subsection*{Scaling Limit}

In this section, we establish conditions ensuring that, after rescaling time and taking the large population limit $(N\to \infty)$,  the vector of frequencies $X^N$ converges to a limit with continuous sample paths. We work under the main assumption that, as the population size grows, the expected size of the potential parent pool decreases, with
\[
    1 - Q_N(\{1\}) =: \rho_N \to 0, \qquad N \to \infty.
\]
The main convergence result, with time rescaled by a factor proportional to \(\rho_N^{-1}\), is stated below.


\begin{theorem}\label{T1SI}
	Let $(C_N)_{N\geq 1}$ be a sequence of coloring rules over $[K]$. 
    For $N\in \mathbb{N}$, consider the frequency process $((X^{N},Y^N))_{N\geq1}$ defined in \eqref{X} and \eqref{Y}  with parameters $(Q_N, q_N, C_N)$. 
    Assume that there exist $\sigma\geq0$, and $\alpha,\beta\geq0$ such that
	\begin{itemize}
		\item[(i)] $\lim_{N\to\infty}\rho_N^{-1}(p_i^N(x)-x_i)=\mu_i(x)\in\mathbb{R}$ for all $i\in[K]$, $x\in\Delta_K$.
		\item[(ii)] $\lim_{N\to\infty} 1/(N\rho_N)=\sigma<\infty$.
		\item[(iii)] $(1-r_N)/\rho_N\to \alpha$ and $q_N/\rho_N\to \beta$.
            \item[(iv)] $\lim_{N \to \infty} \rho_N^{-1} (1 - c_i^N(i)) = 0$ for each $i \in [K]$ and $\sup_{N \geq 1} \sum_{k = 2}^\infty k Q_N(\{k\}) / \rho_N < \infty$.
	\end{itemize}
    Then the sequence $ \left\{\left(X^N( \lfloor t/\rho_N \rfloor ), Y^N( \lfloor t/\rho_N \rfloor )\right) : N \in \mathbb{N} \right\}$
    converges weakly in the space of càdlàg functions $\mathbb{D}(\mathbb{R}_+,\Delta_K \times \Delta_K)$, equipped with the Skorokhod $J_1$ topology, to the unique strong solution $(X(t),Y(t))$ to \eqref{SDEXY}. 
\end{theorem}
Before presenting the proof of Theorem \ref{T1}, we first establish the following auxiliary results.

\begin{lemma} \label{lem1}
    Under the hypotheses of Theorem \ref{T1}, for each $i \in [K]$, $\mu_i$ is Lipschitz with constant 
    \begin{equation} \label{eq:lipsConst} 
        L := 1 + \frac{1}{2} \sup_{N \geq 1} \sum_{k = 2}^\infty k \frac{Q_N(\{k\})}{\rho_N} . 
    \end{equation}
\end{lemma}

\begin{proof}
    First observe that $L$, as defined in \eqref{eq:lipsConst} is finite by assumption (iv) of Theorem \ref{T1}.
    Now fix $i \in [K]$ and $x, y \in \Delta_K$.
    From \eqref{eq:defPi}, for each $N \geq 1$ we get 
    \begin{equation} \label{eq:lemaux1}
    \begin{split}
        \frac{1}{\rho_N} \bigl\lvert (p_i^N(x) - x_i) - (p_i^N(y) - y_i) \bigr\rvert 
        & \leq \frac{Q_N(\{1\})}{\rho_N} \Bigl\lvert \sum_{j = 1}^K c_j^N(i) x_j - x_i - \sum_{j = 1}^K c_j^N(i) y_j + y_i \Bigr\rvert 
        + \sum_{n = 2}^\infty \frac{Q_N(\{n\})}{\rho_N} \lvert x_i - y_i \rvert \\
        & \quad + \sum_{n = 2}^\infty \frac{Q_N(\{n\})}{\rho_N} \Bigl\lvert \sum_{\bm{j} \in [K]^n} \prod_{l = 1}^n x_{jl} c_{\bm{j}}^N(i) - \sum_{\bm{j} \in [K]^n} \prod_{l = 1}^n y_{jl} c_{\bm{j}}^N(i) \Bigr\rvert .
    \end{split}
    \end{equation}
    Observe that because for each $j \in [K]$, $(c_j^N(1), \ldots, c_j^N(K)) \in \Delta_K$ is a probability vector, we obtain 
    \begin{equation} \label{eq:lemaux2}
        \frac{Q_N(\{1\})}{\rho_N} \Bigl\lvert \sum_{j = 1}^K c_j^N(i) x_j - x_i - \sum_{j = 1}^K c_j^N(i) y_j + y_i \Bigr\rvert 
        \leq \sum_{j = 1}^K \frac{1 - c_j^N(j)}{\rho_N} ,
    \end{equation}
    while we directly deduce from $\rho_N = 1 - Q_N(\{1\})$ that 
    \begin{equation} \label{eq:lemaux3}
        \sum_{n = 2}^\infty \frac{Q_N(\{n\})}{\rho_N} \lvert x_i - y_i \rvert = \lvert x_i - y_i \rvert .
    \end{equation}
    With this, we have bounded two of the terms on the right hand side of \eqref{eq:lemaux1}. 
    We will now bound the last term. 

    For the moment fix $N \geq 1$ and $n \geq 1$. 
    Let $J^x = (J_1^x, \ldots, J_n^x)$ and $J^y = (J_1^y, \ldots, J_n^y)$ be random vectors with independent entries such that for each $l \in [n]$ and $j \in [K]$,
    \[
        \mathbb{P}(J_l^x = j) = x_j 
        \quad\text{and}\quad 
        \mathbb{P}(J_l^y = j) = y_j .
    \]
    Considering these random vectors, we can rewrite 
    \[
        \Bigl\lvert \sum_{\bm{j} \in [K]^n} \prod_{l = 1}^n x_{jl} c_{\bm{j}}^N(i) - \sum_{\bm{j} \in [K]^n} \prod_{l = 1}^n y_{jl} c_{\bm{j}}^N(i) \Bigr\rvert
        = \bigl\lvert \mathbb{E}[ c_{J^x}^N(i) - c_{J^y}^N(i) ] \bigr\rvert.
    \]
    Now note that 
    \[
        \bigl\lvert \mathbb{E}[  c_{J^x}^N(i) -  c_{J^y}^N(i) ] \bigr\rvert
        = \bigl\lvert \mathbb{E}[ (c_{J^x}^N(i) - c_{J^y}^N(i)) 1_{\{J^x \neq J^y\}} ] \bigr\rvert
        \leq \mathbb{P}(J^x \neq J^y) 
        \leq \sum_{l = 1}^n \mathbb{P}(J_l^x \neq J_l^y) .
    \]
    By Proposition 4.7 and Remark 4.8 in \cite{levin2017markov}, we can construct $J_l^x$ and $J_l^y$ such that 
    \[
        \mathbb{P}(J_l^x \neq J_l^y) = \lVert x - y \rVert_{\mathrm{TV}} , 
    \]
    where
    \[
        \lVert x - y \rVert_{\mathrm{TV}} := \sup_{A \subset [K]} \Bigl\lvert \sum_{l \in A} (x_l - y_l) \Bigr\rvert
    \]
    is the total variation distance between $x$ and $y$. 
    In this case, Proposition 4.2 of \cite{levin2017markov} gives us the equality 
    \[
        \lVert x - y \rVert_{\mathrm{TV}} = \frac{1}{2} \sum_{l = 1}^K \lvert x_l - y_l \rvert,
    \]
    which lets us obtain the bound
    \begin{equation} \label{eq:lemaux4}
        \Bigl\lvert \sum_{\bm{j} \in [K]^n} \prod_{l = 1}^n x_{jl} c_{\bm{j}}^N(i) - \sum_{\bm{j} \in [K]^n} \prod_{l = 1}^n y_{jl} c_{\bm{j}}^N(i) \Bigr\rvert
        \leq \frac{1}{2} n \sum_{l = 1}^K \lvert x_l - y_l \rvert 
    \end{equation}
    By \eqref{eq:lemaux1}, \eqref{eq:lemaux2}, \eqref{eq:lemaux3}, and \eqref{eq:lemaux4} we deduce that 
    \begin{align*}
        \frac{1}{\rho_N} \bigl\lvert (p_i^N(x) - x_i) - (p_i^N(y) - y_i) \bigr\rvert 
        & \leq \sum_{j = 1}^K \frac{1 - c_j^N(j)}{\rho_N} + \lvert x_i - y_i \rvert
        + \frac{1}{2} \sum_{n = 2}^\infty n \frac{Q_N(\{n\})}{\rho_N} \sum_{l = 1}^K \lvert x_l - y_l \rvert \\
        & \leq \sum_{j = 1}^K \frac{1 - c_j^N(j)}{\rho_N} + L \sum_{l = 1}^K \lvert x_l - y_l \rvert ,
    \end{align*}
    with $L$ as in \eqref{eq:lipsConst}. 
    By condition (i) of Theorem \ref{T1}, we have that 
    \[
        \lim_{N \to \infty} \frac{1}{\rho_N} \bigl\lvert (p_i^N(x) - x_i) - (p_i^N(y) - y_i) \bigr\rvert
        = \lvert \mu_i(x) - \mu_i(y) \rvert ,
    \]
    and by condition (iv) of Theorem \ref{T1}, 
    \[
        \lim_{N \to \infty} \sum_{j = 1}^K \frac{1 - c_j^N(j)}{\rho_N} = 0.
    \]
    Thus, we deduce
    \[
        \lvert \mu_i(x) - \mu_i(y) \rvert \leq L \sum_{l = 1}^K \lvert x_l - y_l \rvert .
    \]
    That is, $\mu_i$ is Lipschitz with constant $L$. 
\end{proof}

\begin{proposition}\label{P1}
	There exists a unique strong solution to the following SDE 
	\begin{equation} \label{SDEXY}
        \begin{split}
		dX(t) & =  \mu(X(t))dt + \sqrt{\sigma} \zeta(X(t)) dB_t +  \alpha (Y(t)-X(t)) dt \\
		dY(t) & = \beta(X(t)-Y(t)) dt \,,
        \end{split}
	\end{equation}
	where $B$ is a $K$-dimensional Brownian motion, $(e_i : 1 \leq i \leq K)$ is the canonical basis of $\mathbb{R}^K$, 
    and for $x \in \Delta_K$:
    \[
        \zeta_{ij}(x) := \begin{dcases*}
            0 & if $i < j$, \\
            \sqrt{\frac{x_i (1 - \sum_{k = 1}^i x_k)}{1 - \sum_{k = 1}^{i - 1} x_k }} & if $i = j$, \\
            - x_i \sqrt{ \frac{x_j}{ (1 - \sum_{k = 1}^{j - 1} x_k) (1 - \sum_{k = 1}^j x_k) } } & if $i > j$.
        \end{dcases*}
    \]

	Moreover, the solution of \eqref{SDEXY} has a generator $\mathcal{A}$ acting on functions $f\in\mathcal{C}^2(\Delta_K\times\Delta_K)$, by
	\begin{align}\label{gen_lim}
		\mathcal{A}f(x,y) & =  \sum_{i = 1}^K (\mu_i(x) + \alpha(y_i - x_i)) \frac{\partial}{\partial x_i} f(x, y) +  \sum_{i = 1}^K \beta(x_i - y_i) \frac{\partial}{\partial y_i} f(x, y) \notag\\
		                  & \quad + \frac{\sigma}{2} \sum_{i, j = 1}^K (1_{\{i = j\}} - x_j) x_i \frac{\partial^2}{\partial x_i \partial x_j} f(x, y).
	\end{align}
\end{proposition}

\begin{proof}
    In order to prove the existence of a unique strong solution to \eqref{SDEXY}, we first establish pathwise uniqueness. \\To this end, we define the $[0,1]^K \times [0,1]^K$-valued process $(U, V)$ by 
\[
    U_i(t) := \sum_{j = 1}^i X_j(t), \quad
    V_i(t) := \sum_{j = 1}^i Y_j(t), \quad
    i \in [K], \ t \ge 0.
\]

Let $W = (W_1, \ldots, W_K)$ be a correlated $K$-dimensional Brownian motion, and define
\[
    \Sigma(x) = (x_1, x_2 - x_1, \ldots, x_K - x_{K - 1}), \quad
    \tilde{\mu}(x) = \mu(\Sigma(x)).
\]

Following the proof of Corollary 3.1 in \cite{casanova2020lambda}, we obtain that, for each $i \in [K]$, the process satisfies the SDE
\begin{equation} \label{sde_u_v}
\begin{split}
    dU_i(t) & = 
    \sum_{j = 1}^i \tilde{\mu}_j(U(t)) \, dt 
               + 
                \alpha  (V_i(t) - U_i(t)) \, dt 
               + \sqrt{\sigma U_i(t) (1 - U_i(t))} \, dW_i(t), \\ 
    dV_i(t) & = 
    \beta (U_i(t) - V_i(t)) \, dt.
\end{split}
\end{equation}
    By Lemma \ref{lem1}, we have that, for $(x,y)\in\Delta_K\times\Delta_K$,
    \[
        \lvert \mu_i(x) - \mu_i(y) \rvert \leq C_i \sum_{j = 1}^K \lvert x_j - y_j \rvert \,,
    \]
    where $C_i$ is a finite constant. 
Therefore, there exist constants $K_1,K_2>0$ such that, for $i\in[K]$ and $u,w\in\Delta_K\times\Delta_K$,
\begin{align*}
\lvert \tilde{\mu_i}(u) - \tilde{\mu_i}(w) \rvert \leq K_1 \sum_{j = 1}^K \lvert u_j - w_j \rvert \,,\notag\\
\lvert \sqrt{u_i(1-u_i)} - \sqrt{w_i(1-w_i)} \rvert^2 \leq K_2\lvert u_i-w_i \rvert.
\end{align*}
    Hence, pathwise uniqueness for \eqref{sde_u_v} follows from Theorem 2 in \cite{Graczyk2013Yamada}; see also Section 4.1 in \cite{casanova2024multitype} for a similar result. Since $X_i = U_i - U_{i-1}$ and $Y_i = V_i - V_{i - 1}$, for $i\in[K]$, we obtain pathwise uniqueness for \eqref{SDEXY}. 
    The existence of a strong solution now follows from Theorem 2.3 and Theorem 1.1 in \cite{IkeWat1981}.
\end{proof}

\begin{proof}[Proof of Theorem \ref{T1SI}]
	We consider the process
	\[
		(\mathcal{X}^N, \mathcal{Y}^N) = \Bigl( X^N( \lfloor \tau^N(t)\rfloor ), Y^N( \lfloor \tau^N(t) \rfloor ) \Bigr)_{t \geq 0} \,,
	\]
	where $\tau^N(t)$ is a Poisson process of intensity $1/\rho_N$.
    By the existence and uniqueness of the solution to \eqref{SDEXY}, it is sufficient to establish the $L^1$-convergence of the sequence of generators $\mathcal{A}^N$, associated with $(\mathcal{X}^N, \mathcal{Y}^N)$, when evaluated at functions $f \in \mathcal{C}^2(\Delta_K\times\Delta_K)$, to the generator $\mathcal{A}$ given in \eqref{gen_lim}, associated to the unique strong solution $(X,Y)$ of Proposition \ref{P1} (see Lemma 17.25 in \cite{kallenberg}). 
    Since the state space $[0,1]$ $\Delta_K \times \Delta_K$ is compact, it is enough to verify pointwise convergence. 
    The desired result then follows from Theorems 17.25 and 17.28 in \cite{kallenberg}.
\medskip

For $f\in\mathcal{C}^2(\Delta_K\times\Delta_K)$ and $(x,y)\in\Delta_K\times\Delta_K$, the infinitesimal generator $\mathcal{A}^N$ of the process $(\mathcal{X}^N, \mathcal{Y}^N)$
    is given by 
	\begin{equation}\label{aux_0}
    \begin{split}
		\mathcal{A}^N f(x, y)
		 & := \lim_{t \downarrow 0} \frac{1}{t} \Bigl( \mathbb{E}\bigl[ f\bigl(X^{N}(\lfloor \tau^N(t) \rfloor), Y^{N}(\lfloor \tau^N(t) \rfloor)\bigr) \bigr] - f(x, y) \Bigr)               \\
		 & = 
          \frac{1}{\rho_N}
         \Bigl( \mathbb{E}\Bigl[ f\Bigl( \frac{M^{(N)}(1)}{N}, Y^{N}(1) \Bigr) \Bigr] - f(x, y) \Bigr) .
    \end{split}
	\end{equation}
	Note that by a Taylor expansion of $f$ we have that
	\begin{equation}\label{aux_1}
    \begin{split}
		\MoveEqLeft
		\frac{\mathbb{E}[ f( M^{(N)}(1) / N, Y^N(1) ) ] - f(x, y) }{ \rho_N }                                                                                                                                                       \\
		 & = \sum_{i = 1}^K \frac{\mathbb{E}[ (M^{(N)}(1))_i - N x_i ]}{N \rho_N } \frac{\partial}{\partial x_i} f(x, y)                                                                                                            \\
		 & \quad + \sum_{i = 1}^K \frac{\mathbb{E}[ Y_i^{(N)}(1) - y_i ]}{\rho_N } \frac{\partial}{\partial y_i} f(x, y)                                                                                                               \\
       & \quad {} + \frac{1}{2} \sum_{j = 1}^K \frac{\mathbb{E}[  (Y_j^{N}(1) - y_j)^2 ]}{N \rho_N } \frac{\partial^2}{\partial y_j^2} f(x, y)
                                                            \\
		 & \quad + \frac{1}{2} \sum_{i, j = 1}^K \frac{\mathbb{E}[ ((M^{(N)}(1))_i - N x_i) ((M^{(N)}(1))_j - N x_j) ]}{N^2 \rho_N } \frac{\partial^2}{\partial x_j \partial x_i} f(x, y)                                        \\
		 & \quad + \frac{1}{2} \sum_{i, j = 1}^K \frac{\mathbb{E}[ ((M^{(N)}(1))_i - N x_i) (Y_j^{N}(1) - y_j) ]}{N \rho_N} \frac{\partial^2}{\partial y_j \partial x_i} f(x, y) + o( 1) \,.
    \end{split}
	\end{equation}
	Observe that, by \eqref{multinomial},
	\begin{align*}
		\frac{\mathbb{E}[ (M^{(N)}(1))_i - N x_i ]}{N \rho_N }
		 & = \frac{p_i^N(r_N x + (1 - r_N) y) - x_i}{\rho_N }                                                                           \\
		 & = \frac{p_i^N(r_N x + (1 - r_N) y) -  r_N x_i - (1 - r_N) y_i}{\rho_N } + \frac{1 - r_N}{\rho_N } (y_i - x_i), \quad i\in[K]\,,
	\end{align*}
	and so, using
	\[
		\frac{p_i^N(x) - x_i}{\rho_N} \to \mu_i(x)\,,\ r_N \to 1\,, \frac{1 - r_N}{\rho_N} \to \alpha
		\text{ as } N \to \infty, \quad i\in[K] \,,
	\]
	it follows that
	\begin{equation} \label{aux1}
		\lim_{N \to \infty} \sum_{i = 1}^K \frac{\mathbb{E}[ (M^{(N)}(1))_i - N x_i ]}{N \rho_N } \frac{\partial}{\partial x_i} f(x, y)
		=  \sum_{i = 1}^K \bigl(\mu_i(x) + \alpha (y_i - x_i)\bigr) \frac{\partial}{\partial x_i} f(x, y)\,.
	\end{equation}
	In a similar fashion, using that
	\[
		\frac{q_N}{\rho_N} \to \beta\ \text{ as }\ N \to \infty,
	\]
	and observing that $\mathbb{E}[Y_i^N(1) - y_i] = q_N (x_i - y_i)$, we deduce
	\begin{align} \label{aux_2}
		\lim_{N \to \infty} \sum_{i = 1}^K \frac{\mathbb{E}[ Y_i^{(N)}(1) - y_i ]}{\rho_N } \frac{\partial}{\partial y_i} f(x, y)
		 & =\lim_{N \to \infty} \sum_{i = 1}^K\frac{q_N}{\rho_N}(x_i-y_i)\frac{\partial}{\partial y_i} f(x, y)\notag \\
		 & = 
         \sum_{i = 1}^K \beta (x_i - y_i) \frac{\partial}{\partial y_i} f(x, y) \,.
	\end{align}
	Using the four previous displays, it also becomes clear that
	\begin{align}\label{aux_3}
		\lim_{N \to \infty} & \sum_{i, j = 1}^K \frac{\mathbb{E}[ ((M^{(N)}(1))_i - N x_i) (Y_j^{N}(1) - y_j) ]}{N \rho_N } \frac{\partial^2}{\partial y_j \partial x_i} f(x, y)\notag \\
		                    & =\lim_{N \to \infty} \sum_{i, j = 1}^K(p_i^N(r_N x + (1 - r_N) y) - x_i)\frac{q_N}{\rho_N}(x_j-y_j)\frac{\partial^2}{\partial y_j \partial x_i} f(x, y)=0
		\,
	\end{align}
	and
	\begin{align}\label{aux_4}
		\lim_{N \to \infty} & \sum_{j = 1}^K \frac{\mathbb{E}[  (Y_j^{N}(1) - y_j)^2 ]}{N \rho_N } \frac{\partial^2}{\partial y_j^2} f(x, y) =\lim_{N \to \infty} \sum_{j = 1}^K\frac{q_N^2}{\rho_N}(x_j-y_j)^2\frac{\partial^2}{\partial y_j^2} f(x, y)=0.
		\,
	\end{align}
	For the sum involving the cross terms of $M^{(1, N)}$, let us first observe that, for $i,j\in[K]$,
\begin{align*}
    \MoveEqLeft
    \mathbb{E}\big[ ((M^{(N)}(1))_i - N x_i) ((M^{(N)}(1))_j - N x_j) \big] \\
    & = \operatorname{Cov}\big( (M^{(N)}(1))_i, (M^{(N)}(1))_j \big) 
      + N^2 (p_i^N(r_N x + (1 - r_N) y) - x_i) (p_j^N(r_N x + (1 - r_N) y) - x_j)\,.
\end{align*}
Hence, \eqref{multinomial} yields
\begin{align}\label{cov_M}
    \MoveEqLeft
    \mathbb{E}\big[ ((M^{(N)}(1))_i - N x_i) ((M^{(N)}(1))_j - N x_j) \big] \notag\\
    & = N \big( p_i^N(r_N x + (1 - r_N) y) 1_{\{i = j\}} 
          - p_i^N(r_N x + (1 - r_N) y) p_j^N(r_N x + (1 - r_N) y) \big) \notag\\
    & \quad + N^2 (p_i^N(r_N x + (1 - r_N) y) - x_i) (p_j^N(r_N x + (1 - r_N) y) - x_j),\qquad i,j\in[K]\,.
\end{align}

Using 
\[
    \frac{1}{N \rho_N } \to \sigma \quad \text{and} \quad 
    p_i^N(r_N x + (1 - r_N) y) \to x_i \quad \text{as } N \to \infty,\qquad i\in[K],
\]
we obtain
\begin{align}\label{aux_5}
    & \lim_{N \to \infty} 
      \frac{p_i^N(r_N x + (1 - r_N) y) 1_{\{i = j\}} - p_i^N(r_N x + (1 - r_N) y) p_j^N(r_N x + (1 - r_N) y)}
           {N \rho_N } = \sigma \, x_i (1_{\{i = j\}} - x_j),\qquad i,j\in[K] \,.
\end{align}

On the other hand, by noting that
\[
    \frac{p_i^N(r_N x + (1 - r_N) y) - x_i}{\rho_N} 
    \to \mu_i(x) + \alpha (y_i - x_i) \quad \text{as } N \to \infty,\qquad i\in[K],
\]
we also deduce
\begin{align}\label{aux_6}
    \lim_{N \to \infty} 
    \frac{(p_i^N(r_N x + (1 - r_N) y) - x_i) (p_j^N(r_N x + (1 - r_N) y) - x_j)}
         {\rho_N } = 0,\qquad i,j\in[K].
\end{align}
Hence, using \eqref{cov_M}, \eqref{aux_5}, and \eqref{aux_6} we obtain
\begin{align}\label{aux_new}
 \lim_{N\to\infty}\frac{1}{2}\sum_{i, j = 1}^K \frac{\mathbb{E}[ ((M^{(N)}(1))_i - N x_i) ((M^{(N)}(1))_j - N x_j) ]}{N^2 \rho_N } \frac{\partial^2}{\partial x_j \partial x_i} f(x, y)=\frac{\sigma}{2}\sum_{i, j = 1}^K \, x_i (1_{\{i = j\}} - x_j)\frac{\partial^2}{\partial x_j \partial x_i} f(x, y).
\end{align}
Letting $N\to\infty$ on both sides of \eqref{aux_1} and using the expressions
from \eqref{aux1}-\eqref{aux_4} together with \eqref{aux_new}, we obtain
	\begin{align}\label{non_extreme_conv}
		\lim_{N \to \infty} \frac{\mathbb{E}[ f( M^{(N)}(1) / N, Y^N(1) ) ] - f(x, y) }{ \rho_N }
		 & = \sum_{i = 1}^K (\mu_i(x) + \alpha (y_i - x_i)) \frac{\partial}{\partial x_i} f(x, y) + \sum_{i = 1}^K \beta (x_i - y_i) \frac{\partial}{\partial y_i} f(x, y) \notag\\
		 & \quad + \frac{\sigma}{2} \sum_{i, j = 1}^K (1_{\{i = j\}} - x_j) x_i \frac{\partial^2}{\partial x_i \partial x_j} f(x, y) \,.
	\end{align}
	Finally, using \eqref{aux_0} and \eqref{non_extreme_conv} we deduce
	\begin{equation*}
		\begin{split}
			\lim_{N\to\infty} \mathcal{A}^N f(x,y) 
			& =  \sum_{i = 1}^K (\mu_i(x) + \alpha (y_i - x_i)) \frac{\partial}{\partial x_i} f(x, y) + \sum_{i = 1}^K \beta (x_i - y_i) \frac{\partial}{\partial y_i} f(x, y) \\
			& \quad + \frac{\sigma}{2} \sum_{i, j = 1}^K (1_{\{i = j\}} - x_j) x_i \frac{\partial^2}{\partial x_i \partial x_j} f(x, y) \\
			&=\mathcal{A}f(x,y),
		\end{split}
	\end{equation*}
	 which completes the proof.
\end{proof}

\subsection*{Tournaments}
In this section, we detail the class of coloring rules introduced in the \emph{Tournaments} section of the main text, which generalize the rock–paper–scissors dynamics observed in side-blotched lizards \cite{sinervo1996rock} to more complex selective interactions (see, e.g., \cite{levine2017beyond}).
\medskip

Consider an odd number of types, $K$, and define an \emph{intransitive} relation $\prec$ on $[K]$, meaning that $i \prec j$ and $j \prec k$ do not necessarily imply $i \prec k$. For simplicity, we assume that each individual has at most two potential parents; 
i.e.,
\(
Q_N(1) + Q_N(2) = 1 
\). 
The coloring rule is given by:
\begin{equation}\label{Tournamentrule}
    c_j^N (k)= c_{jj}^N(k) = 1_{\{k=j\}}
\quad \text{and} \quad
c_{ij}^N(k) = c_{ji}^N(k) = 1_{i \prec j}\,1_{\{k=j\}} + 1_{j \prec i}\,1_{\{k=i\}} \, ,
\end{equation}
for $i, j,k \in [K]$ with $i \neq j$. 
Thus, if an individual has one potential parent, or both parents are of the same type, it inherits that type. If it has two potential parents of different types, it inherits the type that is ``stronger'' under $\prec$.  In Figure~\ref{fig:rpsDiagrams} we show diagrams representing different intransitive relations for $K = 3$, $K = 5$, and $K = 9$. For $K = 3$, the relation  $1 \succ 3$, $3 \succ 2$, and $2 \succ 1$, corresponds to the usual rock-paper-scissors interaction.

\begin{figure}[htbp]
\centering
\subfloat[3RPS]{%
    \includegraphics[width=7.5cm]{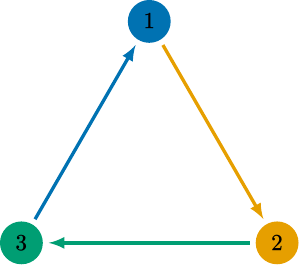}
}
\qquad
\subfloat[5RPS]{%
    \includegraphics[width=7.5cm]{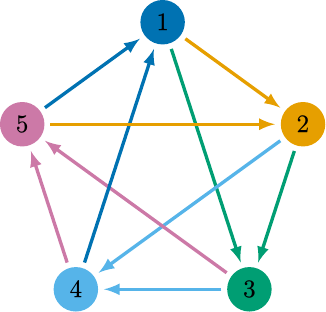}
} \\[2mm]
\subfloat[9RPS]{%
    \includegraphics[width=7.5cm]{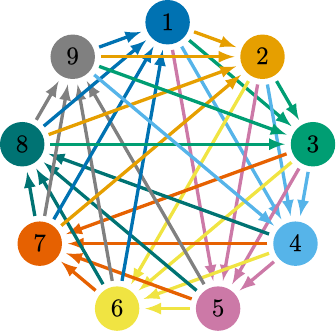}
}
\qquad
\subfloat[Hierarchical RPS]{%
    \includegraphics[width=7.5cm]{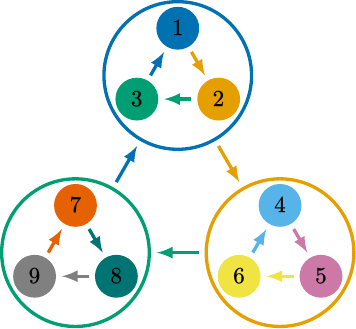}
}
\caption{Diagrams representing different rock-paper-scissors (RPS)–type preference relations. 
Vertices correspond to types, and an arrow from $i$ to $j$ (drawn in the color of vertex $j$) indicates that $i \prec j$, meaning that $j$ dominates $i$ in pairwise interactions. 
Panel (a) shows the usual RPS game with three types. 
Panel (b) shows a system with five types. 
Panels (c) and (d) show two different nine-type systems: (c) follows the cyclic “usual” extension of RPS, while (d) represents a hierarchical dominance structure. 
In all diagrams, edges point to and are colored by the winning type in each interaction.}
\label{fig:rpsDiagrams}
\end{figure}
\newpage

With these rules established, for $i \in [K]$ and $z \in \Delta_K$,
\[p_i^N(z) = Q_N(1)\,z_i + Q_N(2)\Bigl(z_i^2 + 2 \sum_{j \in [K] : j \prec i} z_i z_j\Bigr) \, .\]

Recalling that $\rho_N = 1 - Q_N(1)$, we deduce that
\begin{align*}
\frac{1}{\rho_N}\bigl(p_i^N(z) - z_i\bigr)
&= \frac{1}{\rho_N} \Bigl((1 - \rho_N)z_i + \rho_N\Bigl(z_i^2 + 2 \sum_{j \in [K] : j \prec i} z_i z_j\Bigr) - z_i\Bigr) \\
&= \Bigl(z_i^2 + 2 \sum_{j \in [K] : j \prec i} z_i z_j\Bigr) - z_i \\
&= z_i \Bigl(z_i + 2 \sum_{j \in [K] : j \prec i} z_j - 1\Bigr) \, .
\end{align*}

Noting that
\[
1 = z_i + \sum_{j \in [K] : j \prec i} z_j + \sum_{k \in [K] : i \prec k} z_k ,
\]
we further obtain
\[
\frac{1}{\rho_N}\bigl(p_i^N(z) - z_i\bigr)
= z_i \Bigl( \sum_{j \in [K] : j \prec i} z_j - \sum_{k \in [K] : i \prec k} z_k \Bigr) \, .
\]
Hence, by Theorem \ref{T1}, the drift of the frequency process $X$ defined in \eqref{SDEXY} is given by
\begin{equation*} 
\mu_i(z) = z_i \Bigl( \sum_{j \in [K] : j \prec i} z_j - \sum_{k \in [K] : i \prec k} z_k \Bigr) \, .
\end{equation*}

\subsection*{Majority Rule}
In this section, we provide details of the two-type model introduced in the \textit{Majority rule} section of the main text, and present the proof of the scaling limit given in Eq.~(3) therein. By construction, the model is closely connected to the to the study of hybrid zones \cite{barton1989adaptation} and the Allen-Cahn equation \cite{allen1979microscopic}. In a series of celebrated papers \cite{etheridge2017branching,etheridge2022wide,etheridge2024looking}, a deeper connection between the two concepts was established.
\medskip

We focus on a two-type model with seed bank and the majority rule: the type of an individual is determined by the most common type among its potential parents. More precisely, 
%
%
consider a population of fixed size $N \in \mathbb{N}$ composed of individuals of two types, $a$ and $A$.  Let $\rho_N \in [0,1]$ denote the selection coefficient (fitness advantage) of type $a$. At each generation $n = 1,2,\dots$, the type of each individual is determined according to the following rules:
\begin{itemize}
	\item With probability $1-\rho_N$, the event is neutral: an individual chooses exactly one parent from generation $n-1$ with probability $r_N$, or from generation $n-k  - 1$ with probability $(1-r_N) (1-q_N)^{k-1}q_N$ for $k \ge 1$. The individual inherits the type of the parent.   
    

	\item With probability $\rho_N$, the event is selective and three potential parents are chosen according to the rule above. The individual then adopts the most common type among its potential parents. 
\end{itemize}

Observe that for the type space $K=\{1,2\}$, this corresponds to the case
\[
\mathcal{C}=[K]\cup [K]^3,
\]
with distribution $Q_N$ given by 
$$Q_N(1)=1-\rho_N\ , \ Q_N(3)=\rho_N,$$ and the coloring rule
\begin{equation}\label{MR}
        c^N_j(i) = 1_{\{i=j\}},\qquad c_{(j_1,j_2,j_3)}^N(1) = \begin{cases}
            1& \text{ if } \sum_{i = 1}^3 j_i \le 4\\
            0& \text{ else.}
        \end{cases}
        ,\qquad c_{(j_1,j_2,j_3)}^N(2) = 1-c_{(j_1,j_2,j_3)}^N(1). 
\end{equation}

Now let $u^N(n)$ denote the frequency of type $a$ individuals in generation $n$ (see \eqref{X}), 
and let $y^N(n)$ denote the process tracking the geometrically weighted frequency of 
type $a$ in past generations, defined by (see \eqref{Y}.)
\[
y^N(n):= q_N \sum_{j = 1}^{\infty} (1-q_N)^{j-1} u^N(n-j), 
\qquad n \in \mathbb{Z},
\]
 Building on Theorem \ref{T1}, we derive the system of SDEs that 
describes the scaling limit of the frequency process $u^N$.

\begin{proposition}
For each $N$, let $\rho_N$ be the selection  coefficient of type $a$.  Suppose there exists $s >0$, $\sigma \ge 0$, $\alpha,\beta \ge 0$ so that
\begin{enumerate}
    \item $\lim_{N \to \infty} (N\rho_N)^{-1} = \sigma/s$.
    \item $\lim_{N \to \infty}(1-r_N) / \rho_N = \alpha, $ $ \lim_{N \to \infty}q_N/\rho_N = \beta$.
\end{enumerate}
Then the sequence $\left(u^N\left( \lfloor s t/\rho_N\rfloor\right), y^N\left( \lfloor s t/\rho_N\rfloor\right)\right)$ converges weakly in the space of càdlàg functions $\mathbb{D}(\R_+,[0,1] \times [0, 1])$, equipped with the Skorokhod $J_1$ topology, to the unique strong solution of the following system of SDEs:
\begin{align}\label{ACSI}
    \mathrm du(t) &= s u(t)(1-u(t))(2u(t) - 1) \mathrm dt + \sqrt{\sigma u(t)(1-u(t))}\mathrm dB_t + s \alpha \left(y(t) - u(t) \right)\mathrm dt\notag\\
    \mathrm dy(t)&= s \beta \left(u(t) - y(t) \right)\mathrm dt,
\end{align}
where $B$ is a standard one-dimensional Brownian motion. 
\end{proposition}

\begin{proof}
Let $x\in \Delta_2$. 
    We first consider $\zeta_{ij}$ as in Proposition \ref{P1}:
    \begin{align*}
        \zeta_{12}(x) = \zeta_{22}(x) = 0,\qquad \zeta_{11}(x) =\sqrt{x_1(1-x_1)} 
        ,\qquad \zeta_{22}(x)=-x_2\sqrt{\frac{x_1}{1-x_1}}.
    \end{align*}
    
    Hence, in light of Theorem \ref{T1}, it is enough to verify
    $$p_1^N(x) - x_1 = \rho_N x_1(1-x_1)(2x_1-1). $$
    From the coloring rule definition (\eqref{MR}) we have that:
    \begin{align*}
        p_1^N(x) &= (1-\rho_N)x_1 +\rho_N \sum_{j_1,j_2,j_3 \in \{1,2\}} \prod_{\ell = 1}^3 x_{j_\ell}c^N_{(j_1,j_2,j_3)}(1)=(1-\rho_N)x_1 +\rho_N \left( x_1^3+\sum_{j_1+j_2+j_3 = 4}x_1^2(1-x_1)\right)\\
        &=(1-\rho_N)x_1+\rho_N \left(x_1^3 + 3x_1^2(1-x_1) \right)\\
        &=(1-\rho_N)x_1 + \rho_N\left(3x_1^2-2x_1^3 \right)\\
        &=x_1 + \rho_N\left(3x_1^2-2x_1^3 -x_1 \right)\\
        &=x_1 + \rho_N x_1\left( 1-x_1\right)\left(2x_1 - 1 \right).
    \end{align*}
    On the other hand,
    \begin{align*}
        p_2^N(x) = 1-p_1^N(x) &= x_2 - \rho_N x_1(1-x_1)(2x_1 - 1).
    \end{align*}
\end{proof}

\section*{Codes}
The repository https://doi.org/10.5281/zenodo.20438680 contains the scripts used for all the results presented in this work. More precisely, the results for the \textit{Rock-Paper-Scissors} and \textit{Tournaments} sections were generated using the script {\small\texttt{SDELIMITSIMULATION.R}}.  It begins with the function \texttt{simulate\_SDE}, which implements an Euler-Maruyama scheme for simulating the stochastic differential equation (SDE) given in~\eqref{SDEXY}. A C++ implementation of the same numerical scheme is also provided in \texttt{SDELIMITSIMULATION.cpp}, offering a faster alternative for computing fixation times compared to the \texttt{R} implementation. 
Within \texttt{SDELIMITSIMULATION.R}, the function \texttt{createSigmaOrg} constructs the function $\zeta$ appearing in~\eqref{SDEXY}. The function \texttt{createSigmaAlt} provides a numerically stable implementation of a process whose law coincides with the solution of~\eqref{SDEXY}. 
Two additional functions, \texttt{simulate\_SDE\_Times\_rcpp} and \texttt{run\_SDE\_Times\_vec}, are used to compute fixation times of SDEs of the form~\eqref{SDEXY} under varying parameter configurations.\medskip

Regarding the section \textit{Majority Voting}, the script \texttt{ACSIMULATION.R} contains a complete implementation of an Euler-Maruyama scheme for the SDE described in~\eqref{AC}.

\end{document}